\documentclass{amsart}
\usepackage{epsfig}

\usepackage{url}
\usepackage{algorithm}
\usepackage{algorithmic}
\usepackage{multirow}
\usepackage{amssymb}
\usepackage{graphicx}
\usepackage{epsfig}
\usepackage{amsfonts,amssymb,amscd,amsmath,euscript,enumerate,verbatim,calc}
\usepackage{listings}

\theoremstyle{plain}

\newtheorem{theorem}{Theorem}[section]
\newtheorem{lemma}[theorem]{Lemma}

\theoremstyle{definition}

\newtheorem{example}[theorem]{Example}

\def\slfsr{$\sigma$-LFSR}
\def\GFTwo{${\mathbb F}_{2}$}
\def\GFTwom{${\mathbb F}_{2^m}$}
\def\GFq{${\mathbb F}_q$}

\begin{document}

\title[Maximal periodic Xorshift RNGs]
{Xorshift random number generators from primitive polynomials}
\author{Susil Kumar Bishoi}
\address{Center for Artificial Intelligence and Robotics, \newline \indent Defence Research and Development Organisation, \newline \indent CV Raman Nagar, Bengaluru 560093, India \newline and \newline \indent Faculty of Informatics, \newline \indent Masaryk Univerzity, Czechia}
\email{skbishoi@cair.drdo.in}

\author{Surya Narayan Maharana}
\address{Indian Institute of Technology, \newline \indent  Ropar, Punjab 140001, India}
\email{suryan.math@gmail.com}





\begin{abstract}
A class of xorshift random number generators (RNGs) are introduced by Marsaglia. We have proposed an algorithm which constructs a primitive xorshift RNG from a given primitive polynomial. We also have shown a weakness present in those RNGs and suggested its solution. A separate algorithm also proposed which returns a full periodic xorshift generator with desired number of xorshift operations.

\end{abstract}
\date{\today}
\maketitle

\section{Introduction}\label{intro}
Random bits are required in many areas including in cryptography, computer simulation, statistical sampling, etc. A True Random Number Generator (TRNG) can be used to generate these random bits. However, the TRNG design uses some uncontrollable physical processes as a source of true randomness and in most practical environments this is an inefficient procedure. So, a Pseudo Random Number Generator (PRNG) can be used in place of a TRNG. PRNG takes a small bit length seed (random) as input and produces a very large binary sequence which appears to be random. 
The concept of PRNG motivates the design of stream ciphers and in stream cipher design, Linear Feedback Shift Register (LFSR, see Golomb \cite{golomb}, Lidl and Niederreiter \cite{lidl}) is used as one of the important basic building blocks. 

The LFSR is very popular in hardware as it has fast and low cost of implementation in hardware. If it is primitive, then it produces maximum length periodic bitstream for any nonzero initial state. Also, bitstream generated by the LFSR have very good statistical properties. However, it produces only one new bit in each cycle, so such ciphers are often referred as bit-oriented ciphers and could not take the advantage of available word based modern operations. However, the word based LFSR  called MRMM \cite{GSM,N2,N3,N4} takes this advantage. By Zeng at el., it is called as \slfsr\  \cite{Zeng}. It is shown in \cite{Brent,LECUYER} that Marsaglia's xorshift RNGs are special case of the MRMMs. 

In this paper, we have given an algorithm which constructs xorshift RNGs from a binary primitive polynomial. Later, we have found a weakness in these RNGs generated from this algorithm and suggested a solution to overcome this weakness. The paper is organized as follows. In Section \ref{sec_primLFSR}, we introduce some notations, definitions and results concerning to the primitive LFSRs and xorshift generators. We propose the construction algorithm for xorshift RNGs in Section \ref{sec_construct_xor}. In Section \ref{implementation}, some results and issues pertaining to the xorshift RNGs produced by construction algorithm are discussed. Finally, conclusion is made in Section \ref{conclusion}.

\section {Notation and theory} \label{sec_primLFSR}
Let \GFq\ denote the finite field with $q$ elements, where $q$ is a prime power and \GFq[X] be the ring of polynomials in one variable $X$ with coefficients in \GFq. Denote $M_m$(\GFq) the set of all $m\times m$ matrices with entries in \GFq\ and $GL_m$(\GFq) be the set of all $m \times m$ invertible matrices. 
For $C \in M_m(\mathbb{F}_q)$, $C_{ij}$ denotes the entry of the matrix $C$ at $i^{th}$ row and $j^{th}$ column. For any square matrix $C$, det$(C)=|C|$ denotes its determinant whereas $C^T$ denotes the transpose of the matrix $C$. $ord(C)$ denotes the period of the matrix $C$.
$\lceil n \rceil$ denotes least positive integer greater than or equal to $n$. 
Let $R \in M_m(\mathbb{F}_2)$ be the right shift operator defined as $Rx= (0, x_1,x_2,\hdots, x_{m-1})^T$, where $x=(x_1,x_2,\hdots, x_m)^T \in \mathbb{F}_2^{m}$. Then the matrix form of $R$ is as follows
\[R=\left(%
\begin{array}{ccccc}
  0 & 0 & \cdots & 0 & 0 \\
  1 & 0 & \cdots & 0 & 0 \\
  0 & 1 & \cdots & 0 & 0 \\
  \vdots & \vdots & \ddots & \vdots & \vdots \\
  0 & 0 & \cdots & 1 & 0 \\
\end{array}%
\right)_{m \times m}\]
Similarly, let $L$ be the left-shift operator defined as the transpose of the matrix $R$, i.e., $Lx=L(x_1,x_2,\hdots, x_m)^T = (x_2,x_3,\hdots, x_{m},0)^T$. For a positive integer $k$, $L^k$ means $L$ is applied for $k$ times i.e., $k$ is the amount of shifting in left direction and $R^k$ is defined similarly. It is easy to see that both $L^kx$ and $R^kx=0$ if $k\geq m$. Let $I_m \in GL_m(\mathbb{F}_q)$ be the identity matrix.

\subsection{LFSR} \label{subsec_lfsr}
A sequence $s_0,s_1,s_2,\hdots$ with elements from a finite field \GFq\ is called periodic if there exists a nonnegative integer $p$ such that $s_{i+p}=s_i$ for all $i\geq 0$. 
The smallest such integer $p$ is called the period of the sequence. For a periodic sequence, it is always possible to have a relation called linear recurring relation (LRR) \cite{lidl} 
among the elements as
\begin{equation}\label{eq:lfsr}
 s_{i+n}=-(c_0s_i + c_1s_{i+1} + \cdots + c_{n-1}s_{i+n-1})   
\end{equation}
where $c_i \in $ \GFq\ and the integer $n$ is called the degree of the LRR. It is well known that 
for a given periodic sequence in \GFq\, there is a minimum degree LRR  which satisfy the periodic
sequence. The associated polynomial $f(x) = x^n - c_{n-1} x^{n-1} - \hdots - c_1 x - c_0$ is called
the characteristic polynomial of the LRR.  The companion matrix $T$ of the polynomial $f$ is as follows

\[T=\left(%
\begin{array}{ccccc}
  0 & 1 & 0 & \hdots & 0 \\
  0 & 0 & 1 & \hdots & 0 \\
  \vdots & \vdots & \vdots & \ddots & \vdots\\
  0 & 0 & 0 & \hdots & 1 \\
  c_0 & c_1 & c_2 & \hdots & c_{n-1}
\end{array}%
\right)_{n \times n}\]
If the column vector $S_0=(s_0, s_1,\hdots,s_{n-1})^T \in \mathbb{F}_q^{n}$ is the initial states of the LFSR, then $S_1=T(S_0)=(s_1, s_2,\hdots,s_{n})$ where $s_n$ is calculated using the equation \eqref{eq:lfsr}. The successive states of the LFSR are obtained by repeated application of $T$. If $S_k$ be the states of the LFSR after $k^{th}$ iteration, then $S_k=T^k(S_0)$. Again, it is proved that the sequence generated by the LRR have period ($q^n-1$) if and only if the polynomial associated with the LRR is a primitive polynomial of degree $n$ over the field \GFq\ \cite{lidl,MOV} .

The primitive LFSRs have very nice properties. The primitive LFSRs produce bit sequence which not only have a large period, but also have good statistical properties required for cryptographic applications. Again, they have low cost of implementation in hardware \cite{MOV, Stinson}. So, the LFSRs are quite useful in generation of pseudorandom bit sequences.  However, LFSR prosduce only one new bit in cycle and in many situations such as high speed link encryption, an efficient software encryption technique is required. In such cases, bit-oriented ciphers do not provide adequate efficiency. In case of the LFSR of order $n$, total $n$ shifting along with the feedback computation is needed to produce one bit output. Thus, a single LFSR takes $O(n)$ bit manipulations in order to produce only a single bit. Therefore, in case of software implementation point of view, the LFSR does not take the advantage of the available word based modern processors. However, the generalization of the LFSRs called word based LFSRs (i.e., MMRMs) like xorshift RNGs take this advantage.

\subsection{Xorshift generator } \label{subsec_xorshift}
Xorshift generator \cite{Marsaglia} introduced by Marsaglia is a linear operator $T$, which uses only two word based operations called shifting (both right and left) and exclusive-or (XOR). The basic idea of xorshift generators is that the state is modified by applying repeatedly shift and XOR operations. 
If $\mathbf{S}_0=(\mathbf{s_0}, \mathbf{s_1},\hdots,\mathbf{s_{n-1}})^T \in \mathbb{F}_2^{mn}$ is the initial seed, where each $\mathbf{s_i}$ is $m$-bit in size, then \{$T\mathbf{S}_0, T^2\mathbf{S}_0,T^3\mathbf{S}_0,\hdots$\} is the sequence of words generated by $T$. Note that, in case of xorshift RNGs, $T\mathbf{S}$ can be computed using a small number of xorshift operations  for any $\mathbf{S} \in \mathbb{F}_2^{mn}$. The companion matrix of this operator $T$ in the block form is 

\begin{equation} \label{typeXorShift}
T =
\begin {pmatrix}
\mathbf{0} & I_m & \mathbf{0} & . & \mathbf{0}\\
\mathbf{0} & \mathbf{0} & I_m & . & \mathbf{0}\\
\vdots & \vdots & \vdots & \ddots & \vdots\\
\mathbf{0} & \mathbf{0} & \mathbf{0} & . & I_m\\
(1+L^a)(1+R^b) & \mathbf{0} & \mathbf{0} & . & (1+R^c)
\end {pmatrix}
\end{equation}
where $a,b,c$ are three positive integers and each block is an $m \times m$ matrix. Here $\mathbf{0}$ is the $m \times m$ zero matrix and $T\mathbf{S}_0=(\mathbf{s_1}, \mathbf{s_2}, \hdots,$ $\mathbf{s_{n-1}}, A\mathbf{s_0} + B\mathbf{s_{n-1}})$, where $A=(1+L^a)(1+R^b)$ and $B=(1+R^c)$. So its implementation requires only a few number of xorshift operations per pseudo-random number. Again, the xorshift generators have implementation advantages when the size of the each state in bits is a multiple of the computer word size $m$ (typically $m$ = 32 or 64). The xorshift RNGs are extremely fast and there are several values of triplet $(a,b,c)$ for which the companion matrix $T$ has maximal period. Marsaglia \cite{Marsaglia} lists all those triplets $(a, b, c)$ that yield maximal period xorshift generators with $m=32$ and $m=64$. Later it was verified by Panneton and L'Ecuyer \cite{LECUYER} and also shown some deficiencies after analyzing this class of generators. Brent also discussed a potential problem related to 
correlation of outputs with low Hamming weights and suggested a technique to overcome that problem \cite{Brent}.

From equation-\ref{typeXorShift}, it is clear that  the dimension of the matrix $T$ is $mn \times mn$ and so the maximal period of $T$ could be $2^{mn} - 1$. Let $P(z) = det(T - Iz)$ be the characteristic polynomial of $T$, then $T$ is full periodic (i.e., $ord(T)=2^{mn}-1)$ if and only if $P(z)$ is a primitive polynomial over the binary field \GFTwo \cite{lidl, MOV}. 
The list of triplet $(a,b,c)$ were listed out as follows
\begin{itemize}
 \item It first constructs the matrix $T$ using the triplet$(a,b,c)$ as in equation \eqref{typeXorShift}. 
 \item Checks the primitiveness of the characteristic polynomial $P(z)$ of the matrix $T$.
 \item If $P(z)$ is primitive, then the triplet $(a,b,c)$ is added to the list.
\end{itemize}
In this process, to get one such triplet it needs several attempts for primitiveness checking of the polynomial $P(z)$. Now, we are proposing an algorithm which does the reverse i.e., it first finds a primitive polynomial and then constructs a xorshift generator $T$ from this primitive polynomial. The construction algorithm is described in the following section.  

\section{Construction algorithm for xorshift RNG}\label{sec_construct_xor}
In this section, we present the algorithm which constructs a xorshift RNG from a given primitive polynomial. Let $f(X)=\displaystyle\sum_{i=0}^{mn}a_iX^{i}$ be  a polynomial of degree $mn$ over \GFTwo . Then using the coefficients $a_i$'s define $n$ number of $m \times m$ matrices $C_i$ for $i=1,2,\hdots, n-1$ as below

\begin{equation} \label{C_i}
C_i =
\begin {pmatrix}
0 & 0 & \hdots & 0 & a_i\\
0 & 0 & \hdots & 0 & a_{n+i}\\
0 & 0 & \hdots & 0 & a_{2n+i}\\
\vdots &\vdots &\ddots  &\vdots &\vdots \\ 
0 & 0 & \hdots & 0 & a_{(m-1)n+i}\\
\end {pmatrix}_{m\times m}
\end{equation}
i.e., at least first $(m-1)$ columns are zero columns. Again, define the matrix $C_0$ in the
following form
\begin{equation} \label{C_0}
C_0 =
\begin {pmatrix}
0 & 0 & \hdots & 0 & a_0\\
1 & 0 & \hdots & 0 & a_{n}\\
0 & 1 & \hdots & 0 & a_{2n}\\
\vdots &\vdots &\ddots  &\vdots &\vdots \\ 
0 & 0 & \hdots & 1 & a_{(m-1)n}\\
\end {pmatrix}_{m\times m}
\end{equation}
Now, using the matrix coefficients $C_0, C_1, \hdots, C_{n-1}$, constructs the matrix $T$ of size $mn \times mn$ as given in the equation-\ref{typeMRMM}

\begin{equation} \label{typeMRMM}
T =
\begin {pmatrix}
\mathbf{0} & I_m & \mathbf{0} & . & \mathbf{0}\\
\mathbf{0} & \mathbf{0} & I_m & . & \mathbf{0}\\
\vdots & \vdots & \vdots & \ddots & \vdots\\
\mathbf{0} & \mathbf{0} & \mathbf{0} & . & I_m\\
C_0 & C_1 & C_2 & . & C_{n-1}
\end {pmatrix}
\end{equation}
Then, for the column vector $\mathbf{S}=(\mathbf{s_0}, \mathbf{s_1},\hdots,\mathbf{s_{n-1}})^T$
\begin{eqnarray}\label{eq:efficient_form}
        T{\mathbf S} &=& (\mathbf{s_1}, \mathbf{s_2}, \hdots,\mathbf{s_{n-1}}, C_0\mathbf{s_0} + C_1\mathbf{s_1}+\hdots+ C_{n-1}\mathbf{s_{n-1}})^T 
\end{eqnarray}

Let $M(X)=I_mX^n-C_{n-1}X^{n-1}- \cdots  -C_1X-C_0$. Then $M(X)$ is an $m \times m$ matrix polynomial. We call $M(X)$ as the matrix polynomial corresponding to the polynomial $f(X)$. Using the following results, it is possible to calculate the determinant of an $mn \times mn$  matrix from the determinant of an $m \times m$ matrix \cite [lemma 2.3]{BHH}.

\begin{lemma}\label{dethmatrix1} 
 Let $T$ be the matrix corresponding to the polynomial $f(X)$ of degree $mn$ as defined in \eqref{typeMRMM}. Then the characteristic polynomial of $T$ is equal to the determinant of $M(X)$. 
\end{lemma}
\begin{lemma}\label{dethmatrix2}
 Let $M(X)$ be the matrix polynomial corresponding to the polynomial $f(X)$ of degree $mn$ over \GFq . Then the determinant $|M(X)|$ is equal to $f(X)$. 
\end{lemma}
\begin{proof}
The matrix form of $M(X)$ is 
\begin{equation} \label{hmatrix}
M(X)=
\begin{pmatrix}
  X^n    &    0   &  0   &  \cdots &    0   & f_0 \\
  -1      &    X^n &  0   &  \cdots &    0   & f_1 \\
  0      &    -1   &  X^n &  \cdots &    0   & f_2 \\
  \vdots & \vdots & \vdots & \ddots & \vdots & \vdots   \\
  0      &    0   &  0   &   \cdots &   X^n  & f_{m-2} \\
  0      &    0   &  0   &   \cdots &    -1   & f_{m-1}+f_mX^n 
\end{pmatrix}_{m \times m}
\end{equation}
where $f_i(X)=\displaystyle \sum^{(i+1)n-1}_{k=in}a_kX^{k-in}$ for $i=0,1,\ldots ,(m-1)$ and $f_m(X)=a_{mn}\neq 0$. Multiply $X^n$ with the $n^{\rm th}$ row and add to the $(n-1)^{\rm th}$ row of the above matrix $M(X)$. This will remove $X^n$ from the $(n-1)^{\rm th}$ row without any change in the determinant. Now, add $X^n$ 
times the new $(n-1)^{\rm th}$ row to the $(n-2)^{\rm th}$ row. This will remove $X^n$ from the $(n-2)^{\rm th}$ row. Continue this procedure till all the ${X^n}$ terms on the main diagonal have been removed. Then, the resultant matrix will have the same determinant as $M(X)$ and it will be in the following form
$$\begin{pmatrix}
  0    &    0   &  0   &  \cdots &    0   & g_0  \\
  -1      &    0 &  0  &  \cdots &    0   & g_1 \\
  0      &    -1   &  0 &  \cdots &    0   & g_2 \\
  \vdots & \vdots & \vdots & \ddots & \vdots & \vdots   \\
  0      &    0   &  0   &   \cdots &   0  & g_{m-2} \\
  0      &    0   &  0   &   \cdots &    -1   & g_{m-1}
\end{pmatrix}_{m\times m} $$
where
\begin{eqnarray}
        g_0 & =&  f_0+X^n\left(f_1+ X^n\left(f_2+ \cdots +X^n\left(f_{m-1}+X^nf_m\right)\cdots\right)\right) \nonumber \\
         g_1 &=& f_1+ X^n\left(f_2+ \cdots +X^n\left(f_{m-1}+X^nf_m\right)\cdots\right) \nonumber \\
         g_2 &=& f_2+ \cdots +X^n\left(f_{m-1}+X^nf_m\right) \nonumber \\
         & \vdots & \nonumber \\
         g_{m-2} &=& f_{m-2}+X^n\left(f_{m-1}+X^nf_m\right) \nonumber \\
         g_{m-1} &=& f_{m-1}+X^nf_m \nonumber          
\end{eqnarray}
After suitable operations, it can be shown that 
$$ 
\det\left(M(X)) \right) = 
\det \begin{pmatrix}
    0    &    0   &  0   &  \cdots &    0   & g_0  \\
  -1      &    0 &  0  &  \cdots &    0   & 0 \\
  0      &    -1   &  0 &  \cdots &    0   & 0 \\
  \vdots & \vdots & \vdots & \ddots & \vdots & \vdots   \\
  0      &    0   &  0   &   \cdots &   0  & 0 \\
  0      &    0   &  0   &   \cdots &    -1   & 0
\end{pmatrix} = (-1)^{2(m-1)}g_0.
$$
But $g_0=f(X)$ and thus proves the lemma. 
\end{proof}
From Lemma \ref{dethmatrix1} and Lemma \ref{dethmatrix2}, it is clear that the characteristic polynomial of $T$ is primitive if the polynomial $f(X)$ is primitive. Therefore, if $f(X)$ is primitive, then $T$ is full periodic i.e., $ord(T)=(2^{mn}-1)$ \cite{lidl, MOV}. Our next goal is to show that the matrix operator $T$ belongs to the class of xorshift RNGs.
Let $\mathbf{S}=(\mathbf{s_0}, \mathbf{s_1},\hdots,\mathbf{s_{n-1}})^T \in \mathbb{F}_2^{mn}$, where each $\mathbf{s_i} \in \mathbb{F}_2^{m}$ and then using equation-\ref{typeMRMM} 
\begin{eqnarray}\label{eq:efficient_form}
        T{\mathbf S} &=& (\mathbf{s_1}, \mathbf{s_2}, \hdots,\mathbf{s_{n-1}}, C_0\mathbf{s_0} + C_1\mathbf{s_1}+\hdots+ C_{n-1}\mathbf{s_{n-1}})
\end{eqnarray}
Note that all the matrix coefficients $C_i$ given in equation \eqref{C_i} and \eqref{C_0} used in equation-\ref{eq:efficient_form} have a special form. The construction algorithm for xorshift RNGs takes the advantage of these special structures. It is easy to see that the matrix $C_0$ can be written as $C_0=R+\widehat{C_0}$, where $R$ is the right shift operator and $\widehat{C_0}$ is an $m \times m$ matrix having first $(m-1)$ zero columns. Note that the last column of both $\widehat{C_0}$ and $C_0$ are same and the structure of $\widehat{C_0}$ is exactly same as $C_j$, for $j\geq 1$. 
Because of the following lemma, we will show that $T\mathbf{S}$ can be computed using only xorshift operations.
 \begin{lemma}\cite{BHH}\label{thm:efficient_lemma}
 For any matrix $A \in M_m(\mathbb{F}_2)$ having all the columns zero except the $m^{\rm th}$ column and 
 for any vector ${\mathbf s} = [s_0, s_1, \ldots, s_{m-1}]^T \in \mathbb{F}_2^m$, we have
\[
    A {\mathbf s} =  s_{m-1} {\mathbf v}_m
\]
where ${\mathbf v}_m$ represents the $m^{\rm th}$ column of the matrix $A$.
\end{lemma}

By invoking Lemma \ref{thm:efficient_lemma}, $T\mathbf{S}$ can be rewritten as follows:
\begin{eqnarray*}
        T{\mathbf S} = (\mathbf{s_1}, \mathbf{s_2}, \hdots,\mathbf{s_{n-1}},\mathbf{s_{n}})        
\end{eqnarray*}
where, 
\begin{eqnarray}\label{eq:efficient_form1}
\mathbf{s_{n}}=(R{\mathbf s}_i + \alpha_0 {{\mathbf v}{}_0} + \alpha_1 {\mathbf v}_1 + \cdots + \alpha_{n-1} {\mathbf v}_{n-1}) \label{eqeffi}
\end{eqnarray}
and $\alpha_i$ is the least significant bit (LSB) of ${\mathbf s}_i$ and ${\mathbf v}_i$ is the $m^{\rm th}$ 
column of the matrix $C_i$ ($0\leq i\leq n-1)$, that is ${\mathbf v}_i=[a_i,a_{n+i},\hdots ,a_{(m-1)n+i}]^T$. 
It is clear that the equation \eqref{eq:efficient_form1} can be computed by using only one right shift operation and at most $n$ XOR operations and thus, it falls into the class of Marsaglia's xorshift RNGs. We call equation \eqref{eq:efficient_form1} as feedback computation function.

Now we are in a position to propose the construction algorithm for xorshift RNGs.
The sequential steps of the construction algorithm are described in Algorithm \ref{Algo1}.
\begin{algorithm} 
  \caption{Construction of primitive xorshift RNGs} \label{Algo1}
    \begin{algorithmic}[1]
    \REQUIRE A primitive polynomial $f(X)$ of degree $mn$ over $\mathbb F_2$.
    \ENSURE A xorshift RNG of order $n$ over $\mathbb F_{2^m}$. 
    \STATE Construct the matrix coefficients $C_i$s as in equations \eqref{C_i} and \eqref{C_0}
    \STATE Construct the matrix $T$ as described in the equation-\ref{typeMRMM}
    \STATE Return the matrix $T$
  \end{algorithmic}
\end{algorithm}

The complexity of the Algorithm \ref{Algo1} is O(1) as it generates a primitive xorshift generator $T$ from a given 
primitive polynomial just by expressing the coefficients in matrix form.

\begin{lemma}\label{lemma_tap_points}
 The primitive xorshift RNGs of order $n$ over \GFTwom\ generated by Algorithm \ref{Algo1} will have at least two  and at most $(n+1)$ xorshift operations in the feedback function computation.
\end{lemma}
\begin{proof}
Algorithm \ref{Algo1} generates primitive xorshift RNGs from the primitive polynomial and again, the constant term of the primitive polynomial must be nonzero i.e.,  $a_0\neq 0$. This implies ${\mathbf v}_0 \neq 0$ as ${\mathbf v}_0=[a_0,a_{n},\hdots ,a_{(m-1)n}]^T$. Again, in the recurrence relation \eqref{eqeffi}, the right shift $R$ will be present irrespective of any polynomial. 
So there will be at least two xorshift operations in the feedback computation. Also, in equation \eqref{eqeffi}, there are almost $(n+1)$ nonzero terms. This completes the proof.
\end{proof}

\section{A note on construction algorithm for xorshift generators} \label{implementation}
The xorshift generators generated from the primitive polynomials and so are full periodic. For every primitive polynomial of degree $mn$, it constructs distinct xorshift generator. Therefore, total number of full periodic xorshift generators of order $n$ over the field \GFTwom\ produced by this algorithm is $\frac{\phi(2^{mn}-1)}{mn}$. 
\subsection{Weakness in Initialization of xorshift generator States}
The primitive xorshift RNGs generated by the construction algorithm have efficient software implementation property, however from cryptographic point of view they have a weakness similar to
the Lagged Fibonacci Generator (LFG)\cite{Knuth,Brent2}. In LFG, if all states are initialized with even numbers, then the feedback value will be always even in every iteration. To counter this weakness, at least one state of the LFG must be initialized with an odd value. Similar kind of weakness is also present in the xorshift RNGs constructed by Algorithm \ref{Algo1}. If first $(n-1)$ states (i.e., ${\mathbf s}_0,{\mathbf s}_1,\hdots,{\mathbf s}_{n-2}$)  of the xorshift RNG generated by the construction algorithm are even, then there will be only one active term in the feedback value computation i.e., $R{\mathbf s}_i$. This happens because the $\alpha_i$ defined in equation \eqref{eq:efficient_form1} is the least significant bit (LSB) of the state ${\mathbf s}_i$ and so equal to zero for even value of ${\mathbf s}_i$. If all states are multiple of $2^l$ i.e., ${\mathbf s}_i=2^lk_i$, then there will be only one active component in the feedback function computation till $nl$ many iterations and for $0<j<nl$, 
${\mathbf s}_{n+j}=R^{j_1}{\mathbf s}_{j_2}$, where $j_1=\lceil \frac{j}{n} \rceil$ and $j_2=(j-1)$ mod $n$.

In particular, if the states ${\mathbf s}_i$ for $0\leq i < (n-1)$ are zero vectors and ${\mathbf s}_{n-1}=d$, where  $d$ is a multiple of $2^l$, for some integer $l>0$, then the table-\ref{table_initail_states} gives the states of xorshift generator after the subsequent iterations.

\begin{table}[h] 
  \centering
  \begin{tabular}{|c|c|}
  \hline
  Iteration No. & States of xorshift generator \\ \hline
  1 & $( d,\overbrace{\mathbf{0},\mathbf{0},\cdots,\mathbf{0}}^{\textnormal{(n-1) times}})$\\  \hline
  2 & $( \mathbf{0},d,\overbrace{\mathbf{0},\mathbf{0},\cdots,\mathbf{0}}^{\textnormal{(n-2) times}})$\\  \hline
  3 & $( \mathbf{0},\mathbf{0},d,\overbrace{\mathbf{0},\mathbf{0},\cdots,\mathbf{0}}^{\textnormal{(n-3) times}})$\\  \hline
  \vdots & \vdots\\  \hline
  $n$ & $(\overbrace{\mathbf{0},\mathbf{0},\cdots,\mathbf{0}}^{\textnormal{(n-1) times}},d)$\\ \hline
  $n+1$ & $(\frac{d}{2},\overbrace{\mathbf{0},\mathbf{0},\cdots,\mathbf{0}}^{\textnormal{(n-1) times}})$\\ \hline
  $n+2$ & $(\mathbf{0},\frac{d}{2},\overbrace{\mathbf{0},\mathbf{0},\cdots,\mathbf{0}}^{\textnormal{(n-2)times}})$\\  \hline
    \vdots & \vdots\\  \hline
\end{tabular}\caption{States of xorshift RNG} 
\label{table_initail_states}
\end{table} 

If the content of stage 0 is the output word in each iteration, then the first $nl$ words of
the output sequence produces by the xorshift generator  is as follows
\begin{center}
 $\overbrace{\mathbf{0},\mathbf{0},\cdots,\mathbf{0}}^{\textnormal{(n-1) times}},\dfrac{d}{2},\overbrace{\mathbf{0},\mathbf{0},\cdots,\mathbf{0}}^{\textnormal{(n-1) times}},\dfrac{d}{2^2},\cdots,  
 \overbrace{\mathbf{0},\mathbf{0},\cdots,\mathbf{0}}^{\textnormal{(n-1) times}},\dfrac{d}{2^l}, \cdots$
\end{center}

Here each word is $m-$bit wide. There are $(n-1)$ zero vectors in each $n$ consecutive output words till the $nl^{th}$ iteration. Note that, with initial states $(d,\mathbf{0},\mathbf{0},\cdots,\mathbf{0})$ with $d=2^lk$, the first $ln$ outputs are same irrespective of any primitive xorshift RNGs of order $n$ constructed by Algorithm \ref{Algo1}. So, the initial value of the states of the xorshift generator produced by the construction algorithm are significant for the quality of pseudorandom vectors generation. To avoid this weakness, the initial states of the xorshift RNG need be initialized with odd numbers. In such case, all $\alpha_i$ will be equal to $1$ at the first iteration and there will be maximum number of active terms for the feedback function computation.

\subsection{Different xorshift generators from same binary primitive polynomial}
One of the important thing of the construction algorithm is that it produces different primitive xorshift generators of different order from a given binary primitive polynomial of degree $mn$.
Since in most of the operating system, the word size is of the form $2^k$, we have considered $mn=2^k$ for some positive integer $k$.
Again for $mn=2^k$, there will be $(k-1)$ distinct possible choices for $m$ i.e., $1,2^1,\hdots,2^{k-1}$. For each value of $m$, the construction algorithm returns $n$ 
vectors $\{v_0,v_1,\hdots,v_{n-1}\}$, where each $v_i$ is of $m$-bit length. For better understanding, see the following example.
\begin{example}
 Let us consider the binary primitive polynomial $f(x)=x^{32}+x^{31}+x^{27}+x^{26}+x^{25}+x^{20}+x^{19}+x^{15}+x^{14}+x^{11}+x^{9}+x^{7}+x^{6}+x^{5}+x^{4}+x^{2}+1$. 
 Here degree of $f(x)$ is 32 and so $mn=32=2^5$. Then  the set of possible choices for $m$ is $\{1,2^1,2^2,2^3,2^4\}$. But, we are only considering $m=2^3$ and $2^4$. The respective xorshift RNGs constructed using Algorithm \ref{Algo1} are given below.
 
  \begin{enumerate}
    \item For $m=2^3$ and $n=2^2$:\\$v_0=0$x$f7$, $v_1=0$x$54$, $v_2=0$x$73$, $v_3=0$x$bf$.
    
    \item For $m=2^4$ and $n=2$:\\$v_0= 0$x$bf2f$, $v_1= 0$x$6775$.
  \end{enumerate}
\end{example}

From equation \eqref{eqeffi}, it is clear that for a xorshift generator of order $n$ over the field ${\mathbb F}_{2^{m}}$ requires 
following operations in each iteration
\begin{itemize}
  \item For computation of the feedback value $fd$ as given in \eqref{eq:efficient_form1}, it requires one right shift operation and at most $n$ XOR operations.
  \item $n$ state shifting operations i.e., $\mathbf{s_i}=\mathbf{s_{i+1}}$ for $i=0,1,\hdots,n-2$ and $\mathbf{s_{n-1}}=fd$.
\end{itemize}

Then, using Lemma \ref{lemma_tap_points}, it is clear that at least ($n+2)$ and at most $(2n+1)$ xorshift operations are needed to produce an $m-$bit word in each cycle. 
Thus, to produce a bitstream of length $l$, it will take $\lceil \frac{l}{m}\rceil$ many iterations. If $N$ is the total number of word operations (XOR, right shift and shifting), then $(n+2)\lceil \frac{l}{m}\rceil \leq  N \leq (2n+1)\lceil \frac{l}{m}\rceil$. Suppose for a binary primitive polynomial of degree $mn$, two separate primitive xorshift RNGs (RNG1 and RNG2) are generated with word size $m_1$ and $m_2$ respectively, where $m_2=2m_1$. Let, the respective order of RNG1 and RNG2 be $n_1$ and $n_2$, then $n_1=\frac{mn}{m_1}$ and $n_2=\frac{mn}{m_2}=\frac{mn}{2m_1}$. If $N_1$ and $N_2$ be the total number of operations required to generate bitstream of length $l$, then we have
\begin{eqnarray*}
         (n_1+2)\lceil \frac{l}{m_1}\rceil \leq  &N_1& \leq (2n_1+1)\lceil \frac{l}{m_1}\rceil \\
         (n_2+2)\lceil \frac{l}{m_2}\rceil \leq  &N_2& \leq (2n_2+1)\lceil \frac{l}{m_2}\rceil     
\end{eqnarray*}
Therefore,
\begin{eqnarray}\label{n1_n2_1}
         N_1 &\geq& (n_1+2)\lceil \frac{l}{m_1}\rceil = (2n_2+2)\lceil \frac{2l}{m_2}\rceil \nonumber \\
             &\geq& 2(2n_2+2) (\lceil \frac{l}{m_2}\rceil -1) \nonumber \\
             &=& 2(2n_2+1) \lceil \frac{l}{m_2}\rceil + 2 \lceil \frac{l}{m_2}\rceil  \nonumber  \\
             &>& 2N_2 
\end{eqnarray}
Again,
\begin{eqnarray}\label{n1_n2_2}
         N_1 &\leq& (2n_1+1)\lceil \frac{l}{m_1}\rceil = (4n_2+1)\lceil \frac{2l}{m_2}\rceil \nonumber \\
             &\leq& 2(4n_2+1)\lceil \frac{l}{m_2}\rceil  \nonumber  \\
             &<& 8N_2 
\end{eqnarray}

Using equations \eqref{n1_n2_1} and \eqref{n1_n2_2}, we have $2 < \frac{N_1}{N_2} < 8$. Therefore, for larger value of $m$, the xorshift generator will take lesser number of word operations to produce the bitstream of desired length $l$
and so will take lesser time which is reflected in our experimental results given in the table \ref{table:Performance results}. 
\begin{table}[h]
  \centering
  \begin{tabular}{|c|c|c|}
  \hline
    Word size & Time taken (sec) & Time taken (sec) \\ 
	$m$   & to generate $10^9$ bits & to generate $10^{10}$ bits \\ \hline
    8 &  78.6 & 841 \\  \hline
    16 & 20.02 & 215 \\  \hline
    32 & 5.95 & 62.74 \\  \hline
    64 & 1.94 & 19.17 \\  \hline 
  \end{tabular}\caption{Average timing for different values of $m$}
  \label{table:Performance results}
\end{table} 

For our experiment, we have taken $mn=512$ and the bitstream length $l$ as $10^{9}$ and $10^{10}$. Then for $m=8,16,32,64$, measured the average time taken to generate bitstream of length $l$. It is observed that if the word size $m$ is increased by 2, then the time taken reduced by $c$ to generate a fixed length bitstream, where $2<c<8$. The construction algorithm for primitive xorshift RNGs is implemented in C and the used Test machine is Intel Xeon(R) CPU E5645 @ 2.40GHz x 12 with 8 GiB memory and 64-bit Linux operating system. The table \ref{table:Performance results} summarize the results, which tells that it is better to select the primitive xorshift RNG having larger word size $m$ (i.e., 32 or 64) so as to take the advantage of modern word based operations.

\subsection{Primitive xorshift generator with Desired Number of Tap Points} \label{subsec_tap_pointsNo}
The effect of number of tap points in the LFSR (i.e., the number of nonzero coefficients) is important for cryptographic usage while choosing a primitive polynomial. Because an LFSR with less number of tap positions is susceptible to fast correlation attack \cite{CS,CJS}. The distribution of polynomials over \GFTwo\ with respect to their weights are well studied in \cite{Mishra}. It is desirable to select the primitive polynomial whose weight is close to $\frac{n}{2}$ i.e., the  polynomial is neither too sparse nor too dense \cite{Compagner, LECUYER}. However, in certain areas like light weight cryptography, it is preferable to have less number of nonzero tap positions. So, it is required to have an algorithm which could generate primitive xorshift RNGs of order $n$ over \GFTwom\ with desired number of tap points $k$, where $1<k<n+2$. In case of Marsaglia's xorshift RNGs, there are total six operations (i.e., three XOR and three shifting) are used in the feedback function computation. Algorithm \ref{AlgoTap} 
produces such primitive RNGs with desired number of xorshift operations $k$.

\begin{algorithm}
\caption{Primitive xorshift generator with $k$ xorshift operations} \label{AlgoTap}
\begin{algorithmic}[1]
\REQUIRE Three positive integers $m,n$ and $k$.
\ENSURE A primitive xorshift generator of order $n$ over \GFTwom\ having $k$ xorshift operations.
\STATE Generate a random polynomial $f(X)=\displaystyle\sum_{i=0}^{mn}a_iX^{i}$ as follows
\STATE $i=1$
\WHILE {$i<k$}
\IF{$i=1$}
\STATE $l=0$
\ELSE
\STATE $l = rand$( ) mod $n$ 		/*generating random index*/
\WHILE{$l \in S$}
  \STATE $l = rand$( ) mod $n$
\ENDWHILE
\ENDIF
\STATE $S=S\cup \{l\}$
\STATE $j=1$
\WHILE {$j<m$}
\STATE $a_{l+(j-1)n}= rand$( ) mod 2  	/*generating random bit*/
\STATE $j=j+1$
\ENDWHILE
\STATE $i=i+1$
\ENDWHILE
\STATE $a_0=1$ and $a_{mn}=1$
\STATE if $f(X)$ is not primitive, go to step-1.
\STATE else, using Algorithm-\ref{Algo1} return the required xorshift generator from $f(X)$.
\end{algorithmic}
\end{algorithm}

From lemma-\ref{lemma_tap_points}, it is shown that in case of primitive xorshift RNGs there will be at least two operations for its feedback computation (i.e., $R$ and $\alpha_0{\mathbf v}_0$). Therefore, for getting a primitive xorshift generator with $k$ xorshift operations for its feedback computation, Algorithm \ref{AlgoTap} assigns random binary value to the coefficients needed for the matrix coefficient $C_0$ as given in equation \eqref{C_0} with $a_0=1$. Next it selects  $(k-2)$ distinct random integer $i$ such that $0< i <n$ and then constructs the random binary matrix coefficients $C_i$ as described in equation \eqref{C_i}. Finally, assign $a_{mn}=1$ so that the polynomial $f(X)=\displaystyle\sum_{i=0}^{mn}a_iX^{i}$ will be a polynomial of degree $mn$.
If $f(x)$ is primitive, then Algorithm \ref{Algo1} returns the desired primitive xorshift generator.

\section{Conclusion}\label{conclusion}
In this paper, we have proposed two algorithms related to Marsaglia's xorshift RNGs. Algorithm \ref{Algo1} constructs primitive xorshift generator from a given primitive polynomial. We studied those xorshift generators and found a common weakness in all those generators. It is shown that the states of those xorshift generators need to be initialized carefully and is suggested that all the states to be initialized with odd numbers. We have shown that several primitive xorshift generators of different order  can be constructed from a given primitive polynomial of degree $mn$ using the construction algorithm. We also shown that for the larger value of word size $m$, the xorshift generator takes less time to produce a bitstream of the desired length $l$. So, in case of software implementations, it is suggested to select the primitive xorshift generators with a larger word size $m$ (i.e., 32 or 64) to take the advantage of modern word based operations. Finally, We have provided another algorithm that produces 
efficient primitive xorshift generator with desired number of xorshift operations needed for computation of its feedback function .

\end{document}